\RequirePackage{fix-cm}
\documentclass[smallextended]{svjour3}       
\smartqed  
\usepackage{graphicx}
\usepackage{url}
\usepackage{amsmath}
\usepackage{graphicx}
\usepackage{color}
\newcommand{\leo}{}
\usepackage{tikz}
  
\usepackage{subfig}
\usepackage[text={6.6in,9.6in},centering]{geometry}

  \newcommand{\pn}{{N}} 
  
  \newcommand{\bn}{\mathcal{B}}
  
  \newcommand{\alt}{\textsc{alt}}

\newcommand{\lsa}{\textsc{lsa}} 
\newcommand{\cluster}{\mathcal{C}} 
\newcommand{\ud}[1]{\underline{#1}}  
\newcommand{\po}{\preceq}
\newcommand{\parent}{p}

\newcommand{\rev}{}

\begin{document}

\title{Binets: fundamental building blocks for phylogenetic networks\thanks{Part of this work was conducted while Vincent Moulton was visiting Leo van Iersel on a visitors grant funded by the Netherlands Organization for Scientific Research (NWO). Leo van Iersel was partially supported by NWO, including Vidi grant 639.072.602, and partially by the 4TU Applied Mathematics Institute. \rev{We thank the editor and the two anonymous referees for their constructive comments.}
}}



\author{Leo van Iersel \and Vincent Moulton \and Eveline de Swart \and Taoyang Wu}

\institute{Leo van Iersel \at
Delft Institute of Applied Mathematics\\
Delft University of Technology\\
The Netherlands\\
\email{\url{l.j.j.v.iersel@gmail.com}}
\and Vincent Moulton and Taoyang Wu\at
 School of Computing Sciences\\
 University of East Anglia\\
 Norwich\\
 United Kingdom\\
 \email{\url{v.moulton@uea.ac.uk}}
 \and Eveline de Swart\at
  Delft Institute of Applied Mathematics\\
  Delft University of Technology\\
  The Netherlands\\
  \email{\url{Eveline_de_Swart@hotmail.com}}
  \and Taoyang Wu\at
  School of Computing Sciences\\
 University of East Anglia\\
 Norwich\\
 United Kingdom\\
 \email{\url{Taoyang.Wu@uea.ac.uk}}
}


\maketitle

\begin{abstract}
Phylogenetic networks are a generalization of evolutionary trees that are used by biologists to represent the evolution of organisms which have undergone reticulate evolution. Essentially, a phylogenetic network is a  directed acyclic graph having a unique root in which the leaves are labelled by a given set of species. Recently, some approaches have been developed to construct phylogenetic networks from collections of networks on 2- and 3-leaved networks, which are known as binets and trinets, respectively. Here we study in more depth properties of collections of binets, one of the simplest possible types of networks into which a phylogenetic network can be decomposed.  More specifically, we show that if a collection of level-1 binets is compatible with some binary network, then it is also compatible with a binary level-1 network. Our proofs are based on useful structural results concerning lowest stable ancestors in networks. In addition, we show that, although the binets do not determine the topology of the network, they do determine the number of reticulations in the network, which is one of its most important parameters. We also consider algorithmic questions concerning binets. We show that deciding whether an arbitrary set of binets is compatible with some network is at least as hard as the well-known Graph Isomorphism problem. However, if we restrict to level-1 binets, it is possible to decide in polynomial time whether there exists a binary network that displays all the binets. We also show that to find a network that displays a maximum number of the binets is NP-hard, but that there exists a simple polynomial-time 1/3-approximation algorithm for this problem. It is hoped that these results will eventually assist in the development of new methods for constructing phylogenetic networks from collections of smaller networks.
\end{abstract}

\section{Introduction}

Phylogenetic networks are a generalization of evolutionary trees which biologists use to represent the evolution of species that have undergone reticulate evolution. Such networks are essentially directed acyclic graphs having a unique root in which the leaves are labelled by a set $X$ of species~\cite{hrs}. In contrast to evolutionary trees, which can only represent speciation events, phylogenetic networks permit the representation of evolutionary events such as gene transfer and hybridization which are known to occur in organisms such as bacteria and plants, respectively. Although theoretical properties of evolutionary trees have been studied since at least the 1970's, phylogenetic networks have been considered from this perspective only more recently, especially the rooted variants which we will focus on in this paper. 

One of the most important open questions concerning phylogenetic networks is how to construct them for biological datasets~\cite{bapteste2013networks}. It is now common practice for biologists to construct evolutionary trees from molecular data, and several computer programs are available for this purpose~\cite{felsenstein2004inferring}. However, the problem of constructing networks from such data is an active area of research, and there are only a limited number of programs available for biologists to perform this task. A survey of some of these methods and the theory underpinning phylogenetic networks may be found in~\cite{gus14,hrs,M11}. 

One approach that has been recently developed for constructing phylogenetic networks involves building them up from smaller networks, using what can be thought of as a divide-and-conquer approach~\cite{oldman2016trilonet}. In particular, for a set $X$ of species, a network is constructed for every subset of $X$ size 3 (called a {\em trinet}), and then the trinets are puzzled together to build a network (see Figure~\ref{fig:samebinets} for an example of a trinet). This approach constructs and is based on {\em level-1} networks, networks that are slightly more general than evolutionary trees (see Section~\ref{sec:prelim} for the definition of such networks). 

At first sight, it might appear that trinets are the simplest possible networks that could be considered for building up networks from smaller ones. However, trinets contain even simpler networks called {\em binets}, networks with 2 leaves (see e.g. Figure~\ref{fig:samebinets} for a level-1 trinet and the binets that it displays). Note that whereas binets are the smallest informative building blocks for phylogenetic networks, for rooted phylogenetic trees, these are 3-leaf trees (see e.g.~\cite{byrka2010new}). Interestingly, even though binets are in themselves very simple, the collection of binets displayed by a network can still contain some useful information concerning the network. Indeed, in the aforementioned approach for building level-1 networks from trinets, binets are used in the process of puzzling together the trinets.

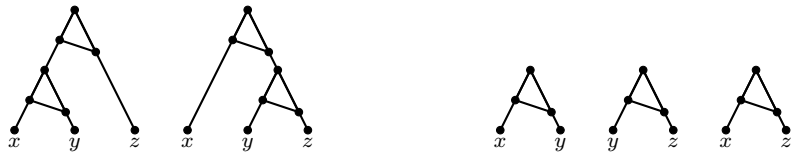
\begin{figure}
	\centering
	\begin{tikzpicture}
	\draw[thick] (0,0) -- (0.4,0.8) -- (0.8,0);
	\draw[thick] (0.4,0.8) -- (0.8,1.6) -- (1.6,0);
	\draw[thick] (0.2,0.4) -- (0.68,0.24);
	\draw[thick] (0.6,1.2) -- (1.08,1.04);
	\draw[thick] (0.2,0.4) -- (0.4,0.8) -- (0.68,0.24);
	\draw[thick] (0.6,1.2) -- (0.8,1.6) -- (1.08,1.04);
	
	\draw [fill] (0.4,0.8) circle (0.05);
	\draw [fill] (0.8,1.6) circle (0.05);
	\draw [fill] (0.2,0.4) circle (0.05);
	\draw [fill] (0.68,0.24) circle (0.05);
	\draw [fill] (0.6,1.2) circle (0.05);
	\draw [fill] (1.08,1.04) circle (0.05);
	\draw [fill] (0,0) circle (0.05) node [below] {$x$};
	\draw [fill] (0.8,0) circle (0.05) node [below] {$y$};
	\draw [fill] (1.6,0) circle (0.05) node [below] {$z$};
	
	\draw[thick] (2.3,0) -- (3.1,1.6) -- (3.5,0.8);
	\draw[thick] (3.1,0) -- (3.5,0.8) -- (3.9,0);
	\draw[thick] (2.9,1.2) -- (3.38,1.04);
	\draw[thick] (3.3,0.4) -- (3.78,0.24);
	\draw[thick] (2.9,1.2) -- (3.1,1.6) -- (3.38,1.04);
	\draw[thick] (3.3,0.4) -- (3.5,0.8) -- (3.78,0.24);
	
	\draw [fill] (3.1,1.6) circle (0.05);
	\draw [fill] (3.5,0.8) circle (0.05);
	\draw [fill] (2.9,1.2) circle (0.05);
	\draw [fill] (3.38,1.04) circle (0.05);
	\draw [fill] (3.3,0.4) circle (0.05);
	\draw [fill] (3.78,0.24) circle (0.05);
	\draw [fill] (2.3,0) circle (0.05) node [below] {$x$};
	\draw [fill] (3.1,0) circle (0.05) node [below] {$y$};
	\draw [fill] (3.9,0) circle (0.05) node [below] {$z$};
	\end{tikzpicture} \hspace{2cm} \begin{tikzpicture}
	\draw[thick] (0,0) -- (0.4,0.8) -- (0.8,0);
	\draw[thick] (0.2,0.4) -- (0.68,0.24);
	\draw[thick] (0.2,0.4) -- (0.4,0.8) -- (0.68,0.24);
	\draw [fill] (0.4,0.8) circle (0.05);
	\draw [fill] (0.68,0.24) circle (0.05);
	\draw [fill] (0.2,0.4) circle (0.05);
	\draw [fill] (0,0) circle (0.05) node [below] {$x$};
	\draw [fill] (0.8,0) circle (0.05) node [below] {$y$};
	\draw[thick] (1.5,0) -- (1.9,0.8) -- (2.3,0);
	\draw[thick] (1.7,0.4) -- (2.18,0.24);
	\draw[thick] (1.7,0.4) -- (1.9,0.8) -- (2.18,0.24);
	\draw [fill] (1.9,0.8) circle (0.05);
	\draw [fill] (2.18,0.24) circle (0.05);
	\draw [fill] (1.7,0.4) circle (0.05);
	\draw [fill] (1.5,0) circle (0.05) node [below] {$y$};
	\draw [fill] (2.3,0) circle (0.05) node [below] {$z$};
	\draw[thick] (3,0) -- (3.4,0.8) -- (3.8,0);
	\draw[thick] (3.2,0.4) -- (3.68,0.24);
	\draw[thick] (3.2,0.4) -- (3.4,0.8) -- (3.68,0.24);
	\draw [fill] (3.4,0.8) circle (0.05);
	\draw [fill] (3.68,0.24) circle (0.05);
	\draw [fill] (3.2,0.4) circle (0.05);
	\draw [fill] (3,0) circle (0.05) node [below] {$x$};
	\draw [fill] (3.8,0) circle (0.05) node [below] {$z$};
	\end{tikzpicture}
	\caption{\label{fig:samebinets}An example of two level-1 \rev{trinets} (left) that display the same set of three binets (right). All arcs are directed downwards. }
\end{figure}

In light of these considerations some obvious questions immediately arise concerning binets. For example, when is a collection of binets displayed by some phylogenetic network (the compatibility problem), and how much information might we expect to extract concerning a phylogenetic network by just looking at the collection of binets that it displays? In this paper, we shall address these and related algorithmic questions concerning binets. It is hoped that these results will be useful in future for developing improved methods for constructing phylogenetic networks from smaller networks.

We now present a summary of the rest of the paper. After introducing some preliminaries concerning phylogenetic networks in the next section, we derive a key structural result for networks (Corollary~\ref{cor:simple:lowest}) which is useful in identifying which of the two possible types of binet is displayed on two leaves within a binary phylogenetic network (that is a network in which all internal vertices have degree 3). Using this theorem, in Section~\ref{sec:binets} we show that the collection of level-1 binets displayed by {\em any} binary phylogenetic network can always be displayed by some binary level-1 network (Theorem~\ref{thm:binary}). This reduces the problem of understanding binets displayed by arbitrary binary networks to level-1 networks. To prove this result, we develop a framework which also implies that there is a polynomial-time algorithm in $|X|$ for deciding whether or not a collection of level-1 binets with combined leaf-set $X$ can be displayed by some network with leaf-set $X$, and, if it is, gives a level-1 network that does this (see Section~\ref{sec:complexity}). Note that this is related to an algorithm presented in~\cite{himsw}. 

In Section~\ref{sec:retic}, we turn to the question as to what can be deduced about the features of a phylogenetic network just by considering the collection of binets that it displays. Note that, as might be expected, there are networks - even trinets - that display the same set of binets but that are not equivalent. For example, the two trinets in Figure~\ref{fig:samebinets} both display the same set of binets, but they are not equivalent. Even so, we will show in Theorem~\ref{thm:retic} that if two level-1 networks both display exactly the same collection of binets, then they must have the same number of reticulation vertices (indegree-2 vertices). Note that the number of such vertices corresponds to the number of reticulate evolutionary events, such as hybridization, that took place in the evolutionary history of the species labelling the leaves of the network. Consequently, the binets displayed by a network can at least capture a useful course-grained feature of the network in question. 

In Sections~\ref{sec:complexity} and~\ref{sec:max}, we consider some algorithmic questions concerning binets. As we have mentioned above, it can be decided in polynomial time in $|X|$ as to when a collection of binets with combined leaf-set $X$ is displayed by some level-1 network on~$X$. However, we show that if we consider arbitrary binets (i.e. not necessarily binary or level-1) then this decision problem becomes at least as hard as the graph-isomorphism problem (see Theorem~\ref{thm:gi}), one of the most famous problems whose complexity is still unknown. In addition, in Section~\ref{sec:max} we consider a related problem which, for a given collection of binary level-1 binets, asks for a network which displays the maximum number of binets in this collection. This is closely related to the maximum rooted triplet consistency problem for evolutionary trees~\cite{byrka2010new}. We show that the binet problem is NP-complete (Theorem~\ref{thm:kBC:hard}), by giving a reduction from the feedback-arc set problem. However, we also show that the problem is 1/3-approximable. In fact, given any collection of binary level-1 binets we can always find some network that displays at least 1/3 of the binets (see Theorem~\ref{thm:questionmark}). We conclude in Section~\ref{sec:discussion} with discussion of some possible future research directions, \rev{and a brief discussion of a potential application of our results.}

\section{Preliminaries}\label{sec:prelim}

Throughout this paper,  $X$ is a non-empty finite set (which usually represents a set of species or organisms). 

\subsection{Digraphs}

A {\em directed graph},  or {\em digraph} for short, $G=(V,E)$ consists of a \rev{finite} set $V=V(G)$ of {\em vertices} and a set $E=E(G)$ of {\em arcs}, where each arc is an ordered pair $(u,v)$ of vertices in $V$ in which $u$ is said to be a {\em parent} of $v$, denoted by $u=\parent(v)$, and $v$ a {\em child} of $u$. All digraphs studied here contain no loops, that is, vertices that are children of themselves. The {\em in-degree} of vertex $u$ is the number of vertices $v$ in $V$ such that $(v,u)$ is an arc, and the {\em out-degree} of $u$ is the number of vertices $w$ with $(u,w)$ being an arc. A \rev{\emph{root}} is a vertex with  in-degree 0.  A {\em leaf} is a vertex of out-degree 0 and the set of leaves is denoted by $L(G)$. 
Any vertex in $G$ that is neither a root nor a leaf is referred to as an {\em interior vertex}. In addition, an interior vertex is a {\em tree vertex} if it has in-degree 1, and a {\em reticulation vertex} if it has in-degree greater than 1. 

 A {\em directed path} or \emph{dipath} in a digraph  is a sequence $u_0,u_1,\ldots,u_k$ ($k\geq 1$) of vertices such that $(u_{i-1},u_i)$ is an arc for $1\leq i \leq k$. An {\em acyclic digraph}  is a digraph that does not contain any directed path starting and ending at the same vertex. If an acyclic digraph $G$ contains a unique root, which is usually designated by $\rho=\rho(G)$, then it will be referred to as a \emph{rooted acyclic digraph}. 
 
 An acyclic digraph $G$ induces a canonical partial order $\prec_{G}$ on its vertex set $V$, that is,  $v\prec_{G} u$ if  there exists a directed path from $u$ to $v$. In this case, we shall say that  $v$ is \emph{below} $u$. \leo{When} the digraph $G$ is clear from the context, $\prec_{G}$ will be written as $\prec$. In addition, we write $v \po u$ if $u=v$ or $u\prec v$. \leo{Given a subset~$U$ of the vertex set of an acyclic digraph, we say that~$u\in U$ is a \emph{lowest} vertex in~$U$ if there is no~$v\in U$ with~$v\prec u$.}

Let $\ud{G}$ be the undirected graph obtained from digraph $G$ by ignoring the direction of the arcs in $G$. Then $G$ is {\em connected} if $\ud{G}$ is connected, that is, there exists an undirected path between every pair of distinct vertices in $\ud{G}$. Note that a rooted acyclic digraph is necessarily connected \rev{(since each connected component of an acyclic digraph has at least one root)}. A {\em cut vertex} is a vertex of $G$ whose removal  disconnects $\ud{G}$. Similarly, a {\em cut arc} is an arc of $G$ whose removal disconnects $\ud{G}$. 
A directed graph is {\em biconnected} if it contains no cut vertex, and a {\em biconnected component} of $G$ is a maximal biconnected subgraph, which is called {\em trivial} if it contains precisely one arc (which is necessarily a cut arc), and {\em non-trivial} otherwise.

\subsection{Phylogenetic networks}

A \emph{phylogenetic network} $\pn$ on $X$ is a rooted acyclic digraph whose leaves are bijectively labeled by the elements in $X$ and which does not contain any vertex with in-degree one and out-degree one. For simplicity, we will just write $L(N)=X$ in case there is no confusion about the labeling. To simplify the argument, throughout this paper we will also assume that all leaves in a phylogenetic network have in-degree one.  
In addition, a phylogenetic network is {\em binary} if each tree vertex, as well as the root, has out-degree 2, and each reticulation vertex has in-degree 2 and out-degree 1. Finally, we say a binary phylogenetic network is \emph{level}-$k$ ($k\geq 0$)  if each of its biconnected components contains at most $k$ reticulation vertices. To some extent, the concept of the level of a phylogenetic network can be regarded as a measure of its `distance' to being a phylogenetic tree. In particular, a binary phylogenetic network is a phylogenetic tree if and only if it is level-0. A phylogenetic network is called {\em simple} if it contains precisely one non-trivial biconnected component $H$ and no \rev{cut arcs} other than the ones leaving $H$.

Two networks $\pn_1=(V_1,E_1)$ and $\pn_2=(V_2,E_2)$ on $X$ are said to be {\em isomorphic} if there exists a bijection $f: V_1\to V_2$ such that  $f(x)=x$ for all $x\in X$, and $(u,v)$ is an arc in $\pn_1$ if and only if $(f(u),f(v))$ is an arc in~$\pn_2$. 

Finally, the \rev{\emph{cluster}} of a vertex $u$, denoted by $\cluster_N(u)=\cluster(u)$, is defined as the subset of~$X$ consisting of the leaves below $u$. Here we will use the convention that $\cluster(u)=\{u\}$ if $u$ is a leaf.

\subsection{Stable ancestors and binets}

Given a phylogenetic network $\pn$ on $X$ and a subset $U\subseteq V(N)$, a {\em stable ancestor} of $U$ in~$N$ is a vertex \rev{$v$} in $V(N)\setminus U$ \rev{such that every path in~$N$ from the root to a vertex in $U$ contains~$v$.} Note that for two stable ancestor\rev{s} $u$ and $u'$ of $U$, we have either $u\po v$ or $v\po u$. Therefore, there exists a unique \leo{lowest vertex in the set of stable ancestors of~$U$}, which will be referred to as the {\em lowest stable ancestor} of $U$ in $N$ and denoted by $\lsa_N(U)=\lsa(U)$. Note that for a subset $Y$ of $X$ with $|Y|\geq 2$, there exist two elements $x$ and $y$ in $Y$ such that $\lsa(Y)=\lsa(\{x,y\})$.
For simplicity, we also write $\lsa(\{x,y\})$ as $\lsa(x,y)$.

The following property of lowest stable ancestors will be useful.

\begin{lemma}
\label{lem:lsa:two:nodes}
Suppose that  $u$ and $v$ are two vertices in a phylogenetic network such that $u\prec v\prec \lsa(u)$, then we have $\lsa(v)\preceq \lsa(u)$.
\end{lemma}
\begin{proof}
Since $u\prec v$, we know that there exists a dipath $P$ from $\rho$ to $u$ that contains~$v$. By the definition of lowest stable ancestor, we know that $\lsa(u)$ and $\lsa(v)$ are contained in $P$. Hence, either $\lsa(v)\preceq \lsa(u)$ or $\lsa(u)\prec \lsa(v)$. If $\lsa(u)\prec \lsa(v)$, then we have $v \prec \lsa(u)\prec \lsa(v)$. Then there exists a dipath $P'$ from $\rho$ to $v$ that does not contain $\lsa(u)$ (otherwise $\lsa(u)$ would be a \rev{stable} ancestor of~$v$ that is below $\lsa(v)$). Using that $u\prec v\prec \lsa(u)$, it follows that there exists a dipath from $\rho$ to $u$ that does not contain $\lsa(u)$, a contradiction. Therefore, $\lsa(v)\preceq \lsa(u)$.
\qed\end{proof}


For~$Y\subseteq X$, the subnet of $\pn$ on $Y$,
 denoted by $\pn|_Y$, is defined as the subgraph obtained from $\pn$ by deleting all vertices that are not on any path from $\lsa(Y)$ to elements in $Y$ and subsequently suppressing all in-degree 1 and out-degree 1 vertices and parallel arcs until no such vertices or arcs exist. A network~$N'$ is said to be \emph{displayed} by network~$N$ if~$N'=N|_Y$ for some~$Y\subseteq X$.
 
Note that, by definition, $\pn|_X=\pn$ if and only if $\lsa(X)=\rho(\pn)$. In this case, $\pn$ is referred to as a  {\em recoverable} network. Note that  every subnet of $\pn$ is necessarily recoverable. \rev{Moreover, a collection of subnets is displayed by some network if and only if it is displayed by some recoverable network. Therefore, we assume all networks in this paper to be recoverable}. 

A {\em binet}  is a phylogenetic network with precisely two leaves, while a {\em trinet}  is a phylogenetic network with precisely three leaves. 
Let
$$\bn(N)=\{N|_Y\,:\,Y\subseteq X~~\mbox{and}~~|Y|=2\}$$ 
be the collection of binets displayed by $N$. Note that there are precisely three binary level-1 binets on a set $\{x,y\}$, and they can be grouped into two types: the ``tree type'', $T(x,y)$, and the ``reticulate type'' $R(x;y)$ and $R(y;x)$  (see Figure~\ref{fig:allBinets}). A collection of binets~$\bn$ \emph{on $X$} is a collection of binets such that the union of the leaf-sets of the binets is equal to~$X$.

\begin{figure}
\centering
 \begin{tikzpicture}
	\draw[thick] (0,0) -- (0.4,0.8) -- (0.8,0);
	\draw [fill] (0.4,0.8) circle (0.05);
	\draw [fill] (0,0) circle (0.05) node [below] {$x$};
	\draw [fill] (0.8,0) circle (0.05) node [below] {$y$};
	\draw (0.4,-1) node {$T(x,y)$};
	\begin{scope}[xshift=2cm,yshift=0cm]
	\draw[thick] (0,0) -- (0.4,0.8) -- (0.8,0);
	\draw[thick] (0.2,0.4) -- (0.68,0.24);
	\draw[thick] (0.2,0.4) -- (0.4,0.8) -- (0.68,0.24);
	\draw [fill] (0.4,0.8) circle (0.05);
	\draw [fill] (0.68,0.24) circle (0.05);
	\draw [fill] (0.2,0.4) circle (0.05);
	\draw [fill] (0,0) circle (0.05) node [below] {$x$};
	\draw [fill] (0.8,0) circle (0.05) node [below] {$y$};
	\draw (0.4,-1) node {$R(x;y)$};
	\end{scope}
	\begin{scope}[xshift=4cm,yshift=0cm]
	\draw[thick] (0,0) -- (0.4,0.8) -- (0.8,0);
	\draw[thick] (0.2,0.4) -- (0.68,0.24);
	\draw[thick] (0.2,0.4) -- (0.4,0.8) -- (0.68,0.24);
	\draw [fill] (0.4,0.8) circle (0.05);
	\draw [fill] (0.68,0.24) circle (0.05);
	\draw [fill] (0.2,0.4) circle (0.05);
	\draw [fill] (0,0) circle (0.05) node [below] {$y$};
	\draw [fill] (0.8,0) circle (0.05) node [below] {$x$};
	\draw (0.4,-1) node {$R(y;x)$};
	\end{scope}
	\begin{scope}[xshift=8cm,yshift=0cm]
	\draw[thick] (0,0) -- (0.4,0.8) -- (0.8,0);
	\draw[thick] (0.2,0.4) -- (0.68,0.24);
	\draw[thick] (0.2,0.4) -- (0.4,0.8) -- (0.68,0.24);
	\draw [fill] (0.4,0.8) circle (0.05);
	\draw [fill] (0.68,0.24) circle (0.05);
	\draw [fill] (0.2,0.4) circle (0.05);
	\draw [fill] (-0.2,-0.4) circle (0.05);
	\draw [fill] (0.8,0) circle (0.05);
	\draw [fill] (1.12,-0.64) circle (0.05) node [below] {$x$};
	\draw [thick] (0.8,0) -- (1.12,-0.64);
	\draw [thick] (0,0) -- (-0.32,-0.64);
	\draw [thick] (0.8,0) -- (-0.2,-0.4);
	\draw [fill] (-0.32,-0.64) circle (0.05) node [below] {$y$};
	\draw (0.4,-1.5) node {\rev{level-2}};
	\end{scope}
	\end{tikzpicture}
\caption{The three binary level-1 binets on $\{x,y\}$ \rev{and an example of a level-2 binet}.}
\label{fig:allBinets}
\end{figure}
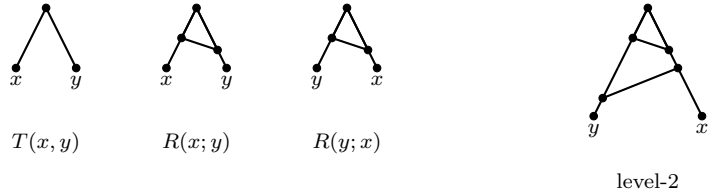

\section{A structure theorem}

In this section we present a key result (Corollary~\ref{cor:simple:lowest}) concerning the structure of the non-trivial biconnected component of a simple network. Note that a similar result has been obtained for a special collection of (non-binary) phylogenetic networks in~\cite{hmw16}.

Let $G$ be a directed acyclic graph and let $P=v_0,v_1,\dots,v_t$ be an undirected path in the underlying undirected graph~$\ud{G}$, then a vertex $v_i$ (with $1\leq i\leq t-1$) is called {\em alternating} (with respect to $P$) if we have either 
$\{(v_{i-1},v_i),(v_{i+1},v_i)\}\subseteq E(G)$ or $\{(v_{i},v_{i-1}),(v_{i},v_{i+1})\}\subseteq E(G)$. The number of alternating vertices contained in $P$ is denoted by $\alt(P)$. Using this concept, we now prove the following \rev{theorem}. \rev{See Figure~\ref{fig:lsa_alt} for an example.}

\begin{figure}
\centering
 \begin{tikzpicture}
	\draw [fill] (0,0) circle (0.05) node [above] {$\rho$};
	\draw [fill] (.5,-.5) circle (0.05) node [above right] {$\lsa(v)$};
	\draw[thick] (0,0) -- (.5,-.5);
	\draw [fill] (-.5,-.5) circle (0.05);
	\draw[ultra thick] (0,0) -- (-.5,-.5);
	\draw [fill] (-1,-1) circle (0.05) node [below left] {$x$};
	\draw[thick] (-1,-1) -- (-.5,-.5);
	\draw [fill] (-.5,-1.5) circle (0.05) node [below left] {$r$};
	\draw[ultra thick] (-.5,-1.5) -- (-.5,-.5);
	\draw [fill] (-.5,-2) circle (0.05) node [below] {$y$};
	\draw[thick] (-.5,-1.5) -- (-.5,-2);
	\draw [fill] (0,-1) circle (0.05);
	\draw [fill] (1,-1) circle (0.05);
	\draw[ultra thick] (0,-1) -- (.5,-.5);
	\draw[thick] (1,-1) -- (.5,-.5);
	\draw[ultra thick] (0,-1) -- (-.5,-1.5);
	\draw [fill] (.5,-1.5) circle (0.05) node [below right] {$v$};
	\draw[thick] (1,-1) -- (.5,-1.5);
	\draw[thick] (0,-1) -- (.5,-1.5);
	\draw [fill] (.5,-2) circle (0.05) node [below] {$q$};
	\draw[thick] (.5,-2) -- (.5,-1.5);
	\draw [fill] (1.5,-1.5) circle (0.05) node [below right] {$p$};
	\draw[thick] (1,-1) -- (1.5,-1.5);
	\end{tikzpicture}
\caption{\rev{Example for the proof of Theorem~\ref{thm:lsa}. The lowest stable ancestor of~$v$ is not equal to the root~$\rho$. An undirected path~$P_u$ between~$\rho$ and $\lsa(v)$ that does not contain the incoming arc of $\lsa(v)$ is indicated in bold. Path~$P_u$ contains one alternating vertex~$r$, for which the lowest stable ancestor is indeed equal to~$\rho$.}}
\label{fig:lsa_alt}
\end{figure}
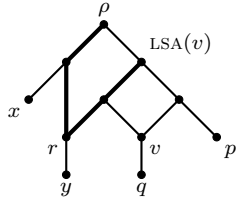

\begin{theorem}\label{thm:lsa}
Let~$N$ be a binary phylogenetic network on~$X$ whose root~$\rho$ is in some non-trivial biconnected component~$H$. Then there exists a lowest vertex in~$H$ with $\lsa(v)=\rho$. 
\end{theorem}
\begin{proof}

Let $\Gamma_0(H)$ be the set of reticulation vertices~$v$ in~$H$ for which the distance \leo{(length of a shortest directed path)} between $\rho$ and $\lsa(v)$ is minimum over all reticulation vertices in~$H$. Note that~$\Gamma_0(H)\neq\emptyset$. 

\medskip

We first show that $\lsa(v)=\rho$ for all $v\in \Gamma_0(H)$. Suppose this were not the case. Then there exists a vertex $v\in \Gamma_0(H)$ such that $\lsa(v)\prec \rho$. Note that~$\lsa(v)$ necessarily has outdegree~2 and therefore has indegree~1
since $N$ is binary and $\lsa(v)\neq \rho$. 
Denote the parent of $\lsa(v)$ by $v^*$. Since~$H$ is biconnected, there exists some undirected path from~$\rho$ to~$\lsa(v)$ that does not contain the edge $e=\{\lsa(v),v^*\}$. Let~$P_u=v_0,\dots,v_t$, where $v_0=\rho$ and $v_t=\lsa(v)$, be such an undirected path for which $\alt(P_u)$ is minimum. 

We claim that~$\alt(P_u)=1$. To see this, note first that since $v_0=\rho$, $v_t=\lsa(v)$ and~$v_{t-1}\neq v^*$, we know that $(v_0,v_1)$ and~$(v_{t},v_{t-1})$ are arcs of~$N$. Hence, $\alt(P_u)$ is odd and strictly positive. Assume for the sake of contradiction that $\alt(P_u)\not =1$, then we have $\alt(P_u)\geq 3$. Let~$v_k$ ($1<k<t$)  be the second alternating vertex contained in $P_u$ \leo{(when travelling from~$v_0$ to~$v_t$)}.

 
  Now fix a directed path~$P_d$ in $N$ from~$\rho$ to~$v_k$.
 
  If the arc $(v^*,\lsa(v))$ is not contained in $P_d$, then, we can find an undirected path from~$\rho$ to~$\lsa(v)$ that does not contain $e$ and has fewer alternating vertices than~$P_u$ by following~$P_d$ until we reach a vertex in~$\{v_k,\dots,v_t\}$ and then following~$P_u$ to~$\lsa(v)$. This gives a contradiction.
  

Now assume that the arc $(v^*,\lsa(v))$ is contained in $P_d$. Then we can find an undirected path from~$\rho$ to~$\lsa(v)$ that does not contain $e$ and has only one alternating vertex as follows. Follow~$P_u$ up to~$v_k$ and then follow~$P_d$ backward from~$v_k$ to~$\lsa(v)$. Since this path has fewer alternating vertices than~$P_u$, we again obtain a contradiction.
  

We have thus shown that~$\alt(P_u)=1$.  Denoting this alternating vertex in $P_u$ by $r$, then $r$ is necessarily a reticulation by the choice of $P_u$. 
 Hence, $P_u$ consists of two directed paths: a directed path from~$\rho$ to~$r$ that does not contain~$\lsa(v)$ and a directed path from~$\lsa(v)$ to~$r$. However, this means that~$\lsa(v)\prec \lsa(r)$, a contradiction to the assumption that~$v\in \Gamma_0(H)$. 

Hence, we know that $\Gamma_0(H)$ is the set of reticulation vertices $v$ of $H$ such that $\lsa(v)=\rho$ and that $\Gamma_0(H)$ is not empty.  

Now fix a vertex $v$ in $\Gamma_0(H)$ that is lowest over all vertices of $\Gamma_0(H)$, that is, there does not exist a vertex $u$ in $\Gamma_0(H)$ such that $v\prec u$.
It remains to show that~$v$ is lowest over all vertices of~$H$. Assume that this is not the case. Then the child~$c$ of~$v$ is also in~$H$. If~$c$ were a reticulation then, by Lemma~\ref{lem:lsa:two:nodes},  $\lsa(v)\preceq \lsa(c)$. However, this would imply that $\lsa(c)=\rho$, contradicting the choice of~$v$. Hence, $c$ is a tree \leo{vertex}.

Since~$H$ is biconnected, there exists some undirected path from~$\rho$ to~$c$ that does not contain $v$. Let~$P_u=w_0,\dots,w_t$ be such a path such that $\alt(P_u)$ is minimum. Note that we have $w_0=\rho$ and $w_t=c$.

Since  $c$ is a tree \leo{vertex} and $P_u$ does not contain its parent~$v$, $(w_{t},w_{t-1})$ is an arc of $N$. Together with $(w_0,w_1)$ being an arc in $N$, we know that $\alt(P_u)$ is odd and strictly positive. We now show, using a similar proof as above, that $\alt(P_u)=1$.
If this were not the case, then we would have $\alt(P_u)\geq 3$. Let~$w_k$ ($1<k<t$)  be the second alternating vertex contained in $P_u$. 
We know that $(w_k,w_{k-1})$ and $(w_k,w_{k+1})$ are two arcs contained in $N$. 
Now fix a directed path~$P_d$ in $N$ from~$\rho$ to~$w_k$.
 
If the vertex $v$ is not contained in $P_d$, then we can find an undirected path from~$\rho$ to~$c$ that does not contain $v$ and has fewer alternating vertices than~$P_u$ by following~$P_d$ from~$\rho$ it reaches a vertex from $\{w_k,\ldots ,w_t\}$ and then following~$P_u$ up to~$c$.  If~$v$ is contained in $P_d$, then we follow~$P_u$ from~$\rho$ to~$w_k$ and then follow~$P_d$ from~$w_k$ to~$c$ and obtain an undirected path from~$\rho$ to~$c$ that does not contain $v$ and has one alternating vertices, which is less than the number of alternating vertices in~$P_u$. In either case, we obtain a contradiction.

We have thus shown that~$\alt(P_u)=1$.  Denoting this alternating vertex in $P_u$ by $r$, then $r$ is necessarily a reticulation by the choice of $P_u$. 
 Hence, $P_u$ consists of two directed paths: a directed path from~$\rho$ to~$r$ that does not contain~$v$ and a directed path from~$c$ to~$r$. However, this means that~$v \prec \lsa(r)$, and hence $\lsa(v)\preceq \lsa(r)$ in view of Lemma~\ref{lem:lsa:two:nodes}. This implies that $r\in \Gamma_0(H)$, a contradiction to the assumption that~$v$ is lowest among $\Gamma_0(H)$.  
\qed\end{proof}

The following is a direct consequence of the above theorem. 

\begin{corollary}
\label{cor:simple:lowest}
Suppose that $N$ is a \leo{simple binary phylogenetic network}.  Let $H$ be the unique non-trivial biconnected component of $N$. Then there exists a lowest vertex $v$ of $H$ such that there exist two arc-disjoint directed paths from the root of $N$ to $v$.
\end{corollary}

\section{Displaying binets by binary networks}\label{sec:binets}

A collection of binary level-1 binets is {\em compatible} if there exists some binary network that displays all binets from the collection. In \rev{this} section, we study the compatibility of binets. Our main result in this section (Theorem~\ref{thm:binary}) shows that when studying the compatibility of binets, we can restrict to binary level-1 networks. 

We will restrict ourselves throughout this section to \emph{thin} collections of binets, i.e. collections containing at most one binet on~$x$ and~$y$ for all distinct~$x,y\in X$. Clearly, any collection of binets that is not thin is not compatible.

First, we need some new definitions. Given a digraph $G$, a {\em sink set} of $G$  is a proper subset $U\subset V(G)$ such that there is no arc leaving $U$, that is, there exists no arc $(x,y)$ with $x\in U$ and $y\in V(G)\setminus U$. A bipartition (or \emph{split}) of $V(G)$ into nonempty sets~$A$ and~$B$, denoted $A|B$, is called
\begin{itemize}
\item \emph{Type I} if both $A$ and $B$ are sink sets (i.e. there is no arc from any element in~$A$ to any element in~$B$ or vice versa);
\item \emph{Type II} if either $A$ or $B$ (but not both) is a sink set; and
\item \emph{Type III} if for all $x\in A,y\in B$ $(x,y)$ is an arc in~$G$ if and only if $(y,x)$ is an arc in $G$.
\end{itemize}
We say that $A|B$ is a \emph{typed} split of $G$ if it is a split of Type I, II or III.

For a collection $\bn$ of binary level-1 binets on~$X$, we introduce the digraph $D(\bn)$ with vertex set $X$ and $(x,y)$ being an arc in $D(\bn)$ if $T(x,y)\in \bn$ or $R(x;y)\in \bn$. \rev{See Figure~\ref{fig:compExample} for an example.}


\rev{The following two lemmas show important properties of typed splits that will be used to establish Theorems~\ref{thm:binet} and~\ref{thm:binary}.}

\begin{lemma}
\label{lem:type}
Suppose that $\bn$ and $\bn'$ are two thin collections of binary level-1 binets on~$X$ with $\bn\subseteq \bn'$. Then each typed split of $D(\bn')$ is a typed split of $D(\bn)$. 
\end{lemma}
\begin{proof}
Suppose that $A|B$ is a typed split of $D(\bn')$. If $A|B$ is of Type I in $D(\bn')$, then it is of Type I in $D(\bn)$ since~$D(\bn)$ is a subgraph of $D(\bn')$. Similarly, if $A|B$ is of Type II in $D(\bn')$, then it is of Type I or II in $D(\bn)$. If $A|B$ is of Type III in $D(\bn')$ then (since $\bn'$ is thin) any binet on~$x$ and~$y$ with $x\in A$ and $y\in B$ is $T(x,y)$. Therefore, $A|B$ is of Type I or III in $D(\bn)$.
\qed\end{proof}

\begin{lemma}
\label{lem:existence:type}
Suppose that $\bn$ is a thin collection of binary level-1 binets on $X$. If $\bn$ is displayed by a binary network, then $D(\bn)$ has a typed split.
\end{lemma}

\begin{proof}
Suppose that $\bn$ is displayed by a binary network. \rev{Then $\bn$ is displayed by a binary recoverable network~$N$}. Let $\bn'$ be the set of binary level-1 binets contained in $\bn(N)$. Then we have $\bn\subseteq \bn' \subseteq \bn(N)$. By Lemma~\ref{lem:type}, it suffices to show that $D(\bn')$ has a typed split.

Consider the root $\rho$ of~$N$, which is equal to~$\lsa(X)$ \rev{since~$N$ is recoverable}. Denote the two children of $\rho$ by $u_1$ and $u_2$. We consider two cases.

The first case is that \rev{at least one arc incident with $\rho$ is a cut arc. Then the other arc incident with $\rho$ is also a cut arc.} Then let $A=\cluster(u_1)$ and $B=\cluster(u_2)$. Note that $A|B$ is a split because neither $A$ nor $B$ is empty. In addition, for all $x\in A, y\in B$ we have $N|_{\{x,y\}}=T(x,y)$ and hence $A|B$ is a Type III split with respect to $D(\bn')$. 

In the second case, both arcs incident with $\rho$ are not cut arcs. Hence, the root $\rho$ is contained in a non-trivial biconnected component $H$ containing~$u_1$ and~$u_2$. By Corollary~\ref{cor:simple:lowest}, there exists a lowest vertex $v$ in $H$ with two arc-disjoint paths~$P_1,P_2$ from $\rho$ to $v$. 
Since $v$ is a lowest vertex in $H$, we know that $v$ is a reticulation vertex and the arc leaving~$v$ is a cut arc.
Let $B=\cluster (v)$ and~$A=X\setminus B$. Then~$B$ is clearly \rev{nonempty}. In addition, $A$ is nonempty, as otherwise $\lsa(X)\preceq v$, a contradiction to the fact that $\lsa(X)=\rho$ (as $N$ is recoverable). Therefore, $A|B$ is a split. 

Consider $x\in A$ and $y\in B$ and \rev{the subnetwork} $N|_{\{x,y\}}$. There is at least one directed path from~$\rho$ to~$x$, and each such path contains at least one arc of~$P_1$ or~$P_2$. Hence, in the process of obtaining $N|_{\{x,y\}}$ from~$N$, the paths~$P_1,P_2$ do not become parallel arcs. Therefore, $N|_{\{x,y\}}$ contains two arc-disjoint paths from~$\rho$ to~$v$ and we can conclude that $N|_{\{x,y\}} \not =T(x,y)$.
 Therefore, if $N|_{\{x,y\}}\in \bn^*$, that is, $N|_{\{x,y\}}$ is level-1, then $N|_{\{x,y\}}=R(x;y)$. This implies that there is no arc $(y,x)$. Therefore, $A|B$ is a Type I or Type II split of $D(\bn^*)$. 
\qed\end{proof}

Note that the condition  that $\bn$ is displayed by a binary network in the above lemma can not be weakened to that $\bn$ is displayed by a network. For example, consider the \rev{binet collection~$\bn$ and network~$N$ in Figure~\ref{fig:compExample}. Although network~$N$ displays~$\bn$, digraph $D(\bn)$ has no typed split (as can be easily checked)}.

\begin{figure}
\centerline{\includegraphics{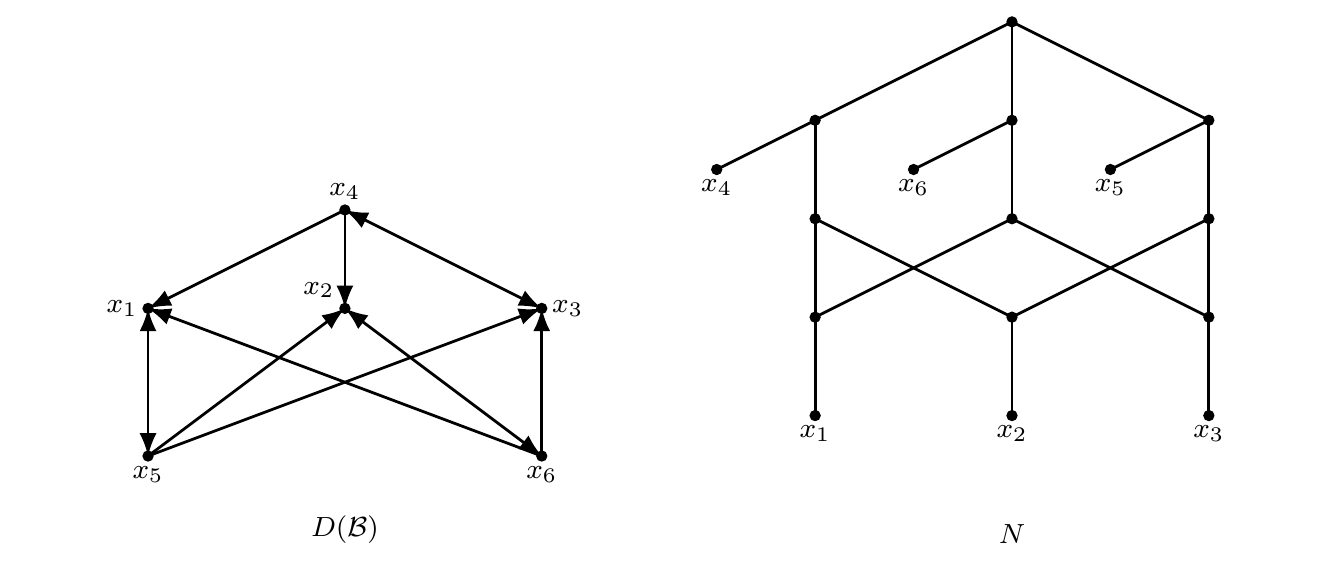}}
\caption{A digraph $D(\bn)$ representing the set of binets~$\bn=\{T(x_4,x_3),R(x_4;x_1),R(x_4;x_2), 
T(x_5,x_1),R(x_5;x_2)$, $R(x_5;x_3), 
T(x_6,x_2),R(x_6;x_1),R(x_6;x_3)\}$ and a network~$N$ displaying~$\bn$. The direction of the arcs in $N$ is downward, and omitted.}
\label{fig:compExample}
\end{figure}

We now introduce two operations, which can be used to combine two phylogenetic networks into a new one. Suppose that $N_1$ and $N_2$ are two phylogenetic networks with disjoint leaf sets. Let $T(N_1,N_2)$ be the phylogenetic network obtained from $N_1$ and $N_2$ by adding a new vertex $v$ and two arcs from $v$ to the roots of $N_1$ and $N_2$. In addition, the network \leo{$R(N_1;N_2)$} is obtained by taking a binet \leo{$R(y_1;y_2)$}, with~$y_1,y_2\notin L(N_1)\cup L(N_2)$, and replacing~$y_i$ by the root of~$N_i$, for~$i=1,2$. \rev{See Figure~\ref{fig:operations} for examples.}

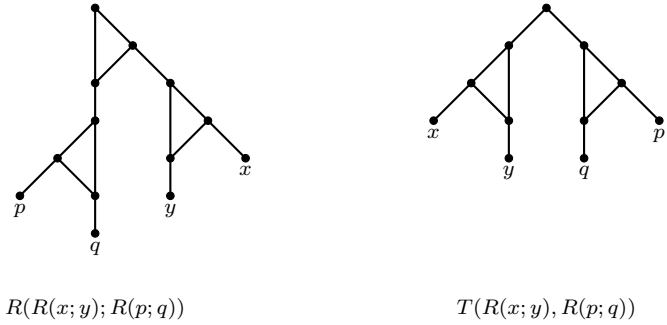
\begin{figure}
\centering
 	\begin{tikzpicture}
	\begin{scope}[xshift=0cm,yshift=0cm]
		\draw [fill] (0,0) circle (0.05);
		\draw [fill] (.5,-.5) circle (0.05);
		\draw [fill] (0,-1) circle (0.05);
		\draw [fill] (1,-1) circle (0.05);
		\draw [fill] (0,-1.5) circle (0.05);
		\draw[thick] (0,0) -- (.5,-.5);
		\draw[thick] (0,0) -- (0,-1);
		\draw[thick] (.5,-.5) -- (0,-1);
		\draw[thick] (.5,-.5) -- (1,-1);
		\draw[thick] (0,-1) -- (0,-1.5);
		\draw [fill] (1,-2) circle (0.05);
		\draw [fill] (1.5,-1.5) circle (0.05);
		\draw[thick] (1,-1) -- (1,-2);
		\draw[thick] (1,-1) -- (1.5,-1.5);
		\draw[thick] (1,-2) -- (1.5,-1.5);
		\draw[thick] (1,-2) -- (1,-2.5);
		\draw [fill] (2,-2) circle (0.05) node[below] {$x$};
		\draw [fill] (1,-2.5) circle (0.05) node[below] {$y$};
		\draw[thick] (1.5,-1.5) -- (2,-2);
		\draw [fill] (0,-2.5) circle (0.05);
		\draw[thick] (0,-1.5) -- (0,-2.5);
		\draw [fill] (-.5,-2) circle (0.05);
		\draw[thick] (0,-1.5) -- (-.5,-2);
		\draw[thick] (0,-2.5) -- (-.5,-2);
		\draw [fill] (0,-3) circle (0.05) node[below] {$q$};
		\draw [fill] (-1,-2.5) circle (0.05) node[below] {$p$};
		\draw[thick] (0,-3) -- (0,-2.5);
		\draw[thick] (-1,-2.5) -- (-.5,-2);
		\draw (0,-4) node {$R(R(x;y);R(p;q))$};
	\end{scope}
	\begin{scope}[xshift=6cm,yshift=0cm]
		\draw [fill] (0,0) circle (0.05);
		\draw [fill] (-.5,-.5) circle (0.05);
		\draw[thick] (0,0) -- (-.5,-.5);
		\draw[thick] (0,0) -- (.5,-.5);
		\draw [fill] (.5,-.5) circle (0.05);
		\draw [fill] (.5,-1.5) circle (0.05);
		\draw [fill] (1,-1) circle (0.05);
		\draw[thick] (.5,-.5) -- (.5,-1.5);
		\draw[thick] (.5,-.5) -- (1,-1);
		\draw[thick] (.5,-1.5) -- (1,-1);
		\draw[thick] (.5,-1.5) -- (.5,-2);
		\draw [fill] (.5,-2) circle (0.05) node[below] {$q$};
		\draw [fill] (1.5,-1.5) circle (0.05) node[below] {$p$};
		\draw[thick] (1.5,-1.5) -- (1,-1);
		\draw [fill] (-.5,-1.5) circle (0.05);
		\draw[thick] (-.5,-1.5) -- (-.5,-.5);
		\draw [fill] (-1,-1) circle (0.05);
		\draw[thick] (-1,-1) -- (-.5,-.5);
		\draw[thick] (-1,-1) -- (-.5,-1.5);
		\draw [fill] (-1.5,-1.5) circle (0.05) node[below] {$x$};
		\draw [fill] (-.5,-2) circle (0.05) node[below] {$y$};
		\draw[thick] (-1.5,-1.5) -- (-1,-1);
		\draw[thick] (-.5,-2) -- (-.5,-1.5);
		\draw (0,-4) node {$T(R(x;y),R(p;q))$};	
	\end{scope}
	\end{tikzpicture}
\caption{\rev{Examples of two networks built recursively by using the operations introduced in the text.}}
\label{fig:operations}
\end{figure}

For a binet set~$\bn$ on~$X$ and a subset~$A\subseteq X$, we define
\[
\bn|_A = \{S\in\bn \,: L(S)\subseteq A\}.
\]

\rev{The next theorem can be used to determine in polynomial time whether a collection of binary level-1 binets is displayed by some binary level-1 network. See Section~\ref{sec:complexity} for more details.}

\begin{theorem}
\label{thm:binet}
Suppose that $\bn$ is a thin collection of binary level-1 binets on $X$. If there exists a typed split $A|B$ of
$D(\bn)$ such that $\bn|_A$ and $\bn|_B$ are both displayed by some binary level-1 network, then $\bn$ is displayed by a binary level-1 network. Moreover, if $\bn$ is displayed by a binary level-1 network, then there exists at least one typed split of $D(\bn)$ and, for each typed split $A|B$ of $D(\bn)$, $\bn|_A$ and $\bn|_B$ are both displayed by some binary level-1 network.
\end{theorem}
\begin{proof}
First suppose that there exists a typed split $A|B$ of $D(\bn)$ such that $\bn|_A$ and $\bn|_B$ are displayed by binary level-1 networks~$N_A$ and~$N_B$, respectively.

If $A|B$ is a Type I or Type III split of $D(\bn)$, then consider the network $N=T(N_A,N_B)$. Then $N$ is a binary level-1 phylogenetic network on $X$ and
 \[
 \bn \subseteq \{T(x,y)\,:x \in A, y\in B\} \cup \bn(N_A) \cup \bn(N_B) =\bn(N),
 \]
 and so $\bn$ is displayed by $N$.

If $A|B$ is a Type II split of $D(\bn)$, then without loss of generality we may assume that $B$ is a sink set in $D(\bn)$. Now consider the network \leo{$N=R(N_A;N_B)$}. 
 Then $N$ is a binary level-1 phylogenetic network on $X$ and
 \[
 \bn\subseteq \{R(x;y)\,:x \in A, y\in B\} \cup \bn(N_A) \cup \bn(N_B) =\bn(N),
 \]
 and so $\bn$ is displayed by $N$.

Now suppose that $\bn$ is displayed by a binary level-1 network~$N$. By 
Lemma~\ref{lem:existence:type}, there exists a typed split $A|B$ of $D(\bn)$. Then $\bn|_A\subseteq \bn(N|_A)$ and $\bn|_B\subseteq \bn(N|_B)$.  
\qed\end{proof}

\rev{We now prove the main result of this section.}

\begin{theorem}
\label{thm:binary}
Suppose that $\bn$ is a thin collection of binary level-1 binets on $X$. Then  $\bn$ is displayed by a binary level-1 network if and only if it is displayed by a binary network.
\end{theorem}
\begin{proof}
Suppose that $\bn$ is displayed by a binary network. We claim that   $\bn$ is also displayed by a binary level-1 network. We shall establish this claim by induction on $|X|$.

If $|X|=2$, then $\bn$ contains at most one binet, which has leaf set $X$. Therefore we know that $\bn$ is displayed by a binary level-1 network.

Now assume that $|X|> 2$, and the claim holds for all sets $X'$ with $2\leq |X'|<|X|$. Let $N$ be a binary network on $X$ with $\bn\subseteq \bn(N)$. By Lemma~\ref{lem:existence:type}, there exists a typed split $A|B$ of $D(\bn)$. Note that  $\bn|_A \subseteq \bn(N|_A)$ and $\bn|_B \subseteq \bn(N|_B)$. Therefore, by induction, each of $\bn|_A$ and $\bn|_B$ is displayed by a binary level-1 network. By Theorem~\ref{thm:binet}, it follows that $\bn$ is displayed by a binary level-1 network. 
\qed\end{proof}

\section{Binets determine the number of reticulations of a binary level-1 network}\label{sec:retic}

In this section we show that, that although the collection of binets displayed by a level-1 network does not necessarily determine the network (see Figure~\ref{fig:samebinets}), it does in fact determine the number of reticulations in the network. We begin by showing that it suffices to consider level-1 networks in which all cycles (in the underlying undirected graph) have length~3.

First, we introduce some further notation. A \emph{semi-cycle}~$C$ of an acyclic directed graph is the union of two non-identical, internally-vertex-disjoint, directed paths from~$s$ to~$t$, with~$s=s(C)$ and~$t=t(C)$ two distinct vertices that are referred to as the {\em source} and {\em terminal} of $C$, respectively.  The \emph{length} of a semi-cycle is the number of distinct vertices that it contains.

\rev{We now show that we may restrict to networks in which all semi-cycles have length~3.}

\begin{lemma}
	\label{lem:cycles}
	If~$N$ is a binary level-1 network, then there exists a binary level-1 network $N'$ in which every semi-cycle has length~3, such that $\bn(N')=\bn(N)$ and~$N$ and~$N'$ have the same number of reticulation vertices.
\end{lemma}
\begin{proof}
	Consider a semi-cycle of~$N$ with source~$s$ and terminal~$t$ and length at least~4. Let~$(u_1,v_1)$, $\ldots$ ,$(u_k,v_k)$, $(t,w)$ be the arcs leaving the semi-cycle. Then~$k\geq 2$.
	Let~$N^*$ be a network obtained from a binary tree on~$\{v_1,\ldots ,v_k\}$ by replacing~$v_i$ by the \leo{subgraph} of~$N$ rooted at~$v_i$, for~$i=1,\ldots ,k$. Let~$N_w$ be the \leo{subgraph} of~$N$ rooted at~$w$.
	Then we construct~$N'$ from~$N$ by replacing the \leo{subgraph} of~$N$ rooted at~$s$ by the network~\leo{$R(N^*;N_w)$}. It is straightforward to see that~$N'$ is a binary level-1 network with the required properties.\qed\end{proof}
	
	\rev{We now establish the main result of this section.}
	
	\begin{theorem}
	\label{thm:retic}
	If~$N_1$ and~$N_2$ are rooted binary level-1 phylogenetic networks on~$X$ with $\bn(N_{1})=\bn(N_{2})$ then~$N_1$ and~$N_2$ have the same number of reticulation vertices.
\end{theorem}
	\begin{proof}
	The proof is by induction on the number of leaves~$|X|$. The induction basis for~$|X|=2$ is clear. Now suppose that~\leo{$N_1$ and~$N_2$ are two non-isomorphic} rooted binary level-1 phylogenetic networks on~$|X|\geq 3$ with $\bn(N_{1})=\bn(N_{2})$ but with different numbers of reticulation vertices. We add an outdegree-1 root to each of~$N_1$ and~$N_2$ with an arc to the original root. By Lemma~\ref{lem:cycles}, we may assume that all semi-cycles in~$N_1$ and~$N_2$ have length~3.
	
	Choose an arbitrary leaf~$x\in X$ and let~$X'=X\setminus \{x\}$. Let $N_1'$ and $N_2'$ be the networks obtained from $N_1|_{X'}$ and $N_2|_{X'}$, respectively, by adding an outdegree-1 root with an arc to the original root. Then~$N_1'$ and~$N_2'$ have the same number of reticulation vertices by induction. 
	
	Since all semi-cycles in~$N_1$ and~$N_2$ are assumed to have length~3, there are three cases for the location of~$x$ in each of the networks~$N_1,N_2$, illustrated in Figure~\ref{fig:addingx}.
	
	If the parent of~$x$ is in a semi-cycle in~$N_i$, let~$v_i$ be the source of this semi-cycle, and let~$v_i$ be the parent of~$x$ otherwise. Let~$B_i := \cluster(v_i)\setminus \{x\}$ and $A_i := X'\setminus B_i$ \rev{(recall that $\cluster(v_i)$ denotes the cluster of~$v_i$)}.
	
		We now consider the different ways in which we could add~$x$ to both networks. Since~$N_1$ and~$N_2$ have different numbers of reticulation vertices, there are two cases to consider (after eliminating symmetric cases), as illustrated in Figure~\ref{fig:2cases}.
	
	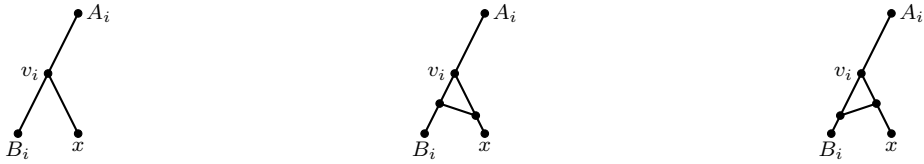
\begin{figure}
\centering\subfloat
		[The parent of~$x$ is not in a semi-cycle.]
		{
		\hspace{1.5cm}\begin{tikzpicture}[baseline={([yshift=-1em] current bounding box.north)}]
		\draw[thick] (0,0) -- (0.4,0.8) -- (0.8,0);
		\draw[thick] (0.4,0.8) -- (0.8, 1.6);
		\draw [fill] (0,0) circle (0.05) node [below] {$B_i$};
		\draw [fill] (0.4,0.8) circle (0.05) node [left] {$v_i$};
		\draw [fill] (0.8,0) circle (0.05) node [below] {$x$};
		\draw [fill] (0.8,1.6) circle (0.05) node [right] {$A_i$};
		\end{tikzpicture}\hspace{1.5cm}
		}\qquad
		\subfloat
		[The parent of~$x$ is the terminal of a semi-cycle.]
		{
		\hspace{1.5cm}\begin{tikzpicture}[baseline={([yshift=-1em] current bounding box.north)}]
		\draw[thick] (0,0) -- (0.4,0.8) -- (0.8,0);
		\draw[thick] (0.4,0.8) -- (0.8, 1.6);
		\draw[thick] (0.2,0.4) -- (0.68,0.24);	
		\draw [fill] (0.2,0.4) circle (0.05);		
		\draw [fill] (0.68,0.24) circle (0.05);					
		\draw [fill] (0,0) circle (0.05) node [below] {$B_i$};
		\draw [fill] (0.4,0.8) circle (0.05) node [left] {$v_i$};
		\draw [fill] (0.8,0) circle (0.05) node [below] {$x$};
		\draw [fill] (0.8,1.6) circle (0.05) node [right] {$A_i$};
		\end{tikzpicture}\hspace{1.5cm}
		}\qquad
		\subfloat[The parent of~$x$ is a non-terminal non-source vertex of a semi-cycle.]
		{
		\hspace{1.5cm}\begin{tikzpicture}[baseline={([yshift=-1em] current bounding box.north)}]
		\draw[thick] (0,0) -- (0.4,0.8) -- (0.8,0);
		\draw[thick] (0.4,0.8) -- (0.8, 1.6);
		\draw[thick] (0.6,0.4) -- (0.12,0.24);
		\draw [fill] (0.6,0.4) circle (0.05);
		\draw [fill] (0.12,0.24) circle (0.05);
		\draw [fill] (0,0) circle (0.05) node [below] {$B_i$};
		\draw [fill] (0.4,0.8) circle (0.05) node [left] {$v_i$};
		\draw [fill] (0.8,0) circle (0.05) node [below] {$x$};
		\draw [fill] (0.8,1.6) circle (0.05) node [right] {$A_i$};
		\end{tikzpicture}\hspace{1.5cm}
		}
		\caption{\label{fig:addingx}The three cases for the location of~$x$ in each of the networks~$N_1$ and~$N_2$ in the proof of Theorem~\ref{thm:retic}.}
		\end{figure}
	
	\begin{figure}
	\centering
    		\subfloat[Case 1: the parent of $x$ is not in a semi-cycle in~$N_1$ \rev{but is} the terminal of a semi-cycle in~$N_2$.]{
		\begin{tikzpicture}[baseline={([yshift=-1em] current bounding box.north)}]
		\draw[thick] (0,0) -- (0.4,0.8) -- (0.8,0);
		\draw[thick] (0.4,0.8) -- (0.8, 1.6);
		\draw [fill] (0,0) circle (0.05) node [below] {$B_1$};
		\draw [fill] (0.4,0.8) circle (0.05) node [left] {$v_1$};
		\draw [fill] (0.8,0) circle (0.05) node [below] {$x$};
		\draw [fill] (0.8,1.6) circle (0.05) node [right] {$A_1$};
		\draw[thick] (0.4,-1) node {$N_1$};
		\end{tikzpicture}
		\hspace{1cm}\begin{tikzpicture}[baseline={([yshift=-1em] current bounding box.north)}]
		\draw[thick] (0,0) -- (0.4,0.8) -- (0.8,0);
		\draw[thick] (0.4,0.8) -- (0.8, 1.6);
		\draw[thick] (0.2,0.4) -- (0.68,0.24);			
		\draw [fill] (0.2,0.4) circle (0.05);
		\draw [fill] (0.68,0.24) circle (0.05);
		\draw [fill] (0,0) circle (0.05) node [below] {$B_2$};
		\draw [fill] (0.4,0.8) circle (0.05) node [left] {$v_2$};
		\draw [fill] (0.8,0) circle (0.05) node [below] {$x$};
		\draw [fill] (0.8,1.6) circle (0.05) node [right] {$A_2$};
		\draw[thick] (0.4,-1) node {$N_2$};
		\end{tikzpicture}}
		\hspace{2cm}
    		\subfloat[Case 2: the parent of $x$ is not in a semi-cycle in~$N_1$ \rev{but is} the non-terminal non-source vertex of a semi-cycle in~$N_2$.]{
		\begin{tikzpicture}[baseline={([yshift=-1em] current bounding box.north)}]
		\draw[thick] (0,0) -- (0.4,0.8) -- (0.8,0);
		\draw[thick] (0.4,0.8) -- (0.8, 1.6);	
		\draw [fill] (0,0) circle (0.05) node [below] {$B_1$};
		\draw [fill] (0.4,0.8) circle (0.05) node [left] {$v_1$};
		\draw [fill] (0.8,0) circle (0.05) node [below] {$x$};
		\draw [fill] (0.8,1.6) circle (0.05) node [right] {$A_1$};
				\draw[thick] (0.4,-1) node {$N_1$};
		\end{tikzpicture}\hspace{1cm}\begin{tikzpicture}[baseline={([yshift=-1em] current bounding box.north)}]
		\draw[thick] (0,0) -- (0.4,0.8) -- (0.8,0);
		\draw[thick] (0.4,0.8) -- (0.8, 1.6);
		\draw[thick] (0.6,0.4) -- (0.12,0.24);	
		\draw [fill] (0.6,0.4) circle (0.05);
		\draw [fill] (0.12,0.24) circle (0.05);	
		\draw [fill] (0,0) circle (0.05) node [below] {$B_2$};
		\draw [fill] (0.4,0.8) circle (0.05) node [left] {$v_2$};
		\draw [fill] (0.8,0) circle (0.05) node [below] {$x$};
		\draw [fill] (0.8,1.6) circle (0.05) node [right] {$A_2$};
				\draw[thick] (0.4,-1) node {$N_2$};
		\end{tikzpicture}	}
		\caption{\label{fig:2cases} The two possible ways to add leaf~$x$ to~$N_1'$ and~$N_2'$ in the proof of Theorem~\ref{thm:retic} such that the obtained networks~$N_1$ and~$N_2$ have different numbers of reticulation vertices.}
	\end{figure}
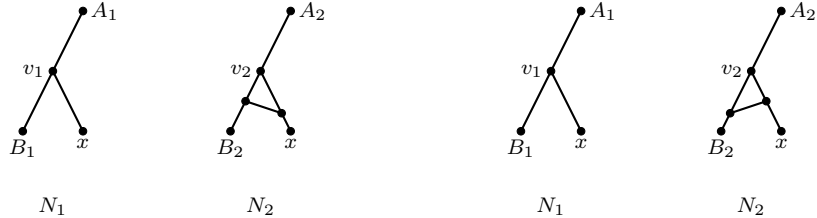
	
	The first case is that the parent of $x$ is not in a semi-cycle in~$N_1$ \rev{but is} the terminal of a semi-cycle in~$N_2$. First suppose that~$B_1\cap B_2 \neq\emptyset$. Then choose an arbitrary vertex~$y\in B_1\cap B_2$. Then $N_1|_{\{x,y\}}=T(x,y)$ while $N_2|_{\{x,y\}}=R(y;x)$, a contradiction. Hence, we may assume that $B_1\cap B_2 =\emptyset$. Then $B_1=A_2$ and $B_2=A_1$. Clearly, $B_1,B_2\neq\emptyset$. Take~$y\in B_1=A_2$ and~$z\in B_2=A_1$. Then $N_1|_{\{x,y\}}=T(x,y)$ and hence $N_2|_{\{x,y\}}=T(x,y)$, from which we can deduce that $N_2|_{\{z,y\}}=T(z,y)$. In addition, $N_2|_{\{z,x\}}=R(z;x)$ and hence $N_1|_{\{z,x\}}=R(z;x)$, from which we can deduce that $N_1|_{\{z,y\}}=R(z;y)$. This leads to a contradiction since $N_2|_{\{z,y\}}=T(z,y)$.
	
	The second case is that the parent of $x$ is not in a semi-cycle in~$N_1$ \rev{but is} the non-terminal non-source vertex of a semi-cycle in~$N_2$. First suppose that~$B_1\cap B_2 \neq\emptyset$. Then choose an arbitrary vertex~$y\in B_1\cap B_2$. Then $N_1|_{\{x,y\}}=T(x,y)$ while $N_2|_{\{x,y\}}=R(x;y)$, a contradiction. Hence, we may assume that $B_1\cap B_2 =\emptyset$. Then, as in the previous case, $B_1=A_2\neq\emptyset$ and $A_1=B_2\neq\emptyset$. Take~$y\in B_1=A_2$ and~$z\in B_2=A_1$. Then, similar to the previous case, $N_1|_{\{x,y\}}=T(x,y)$ and hence $N_2|_{\{x,y\}}=T(x,y)$, from which we can deduce that $N_2|_{\{z,y\}}=T(z,y)$. In addition, $N_2|_{\{z,x\}}=R(x;z)$ and hence $N_1|_{\{z,x\}}=R(x;z)$, from which we can deduce that $N_1|_{\{z,y\}}=R(y;z)$. This again leads to a contradiction since $N_2|_{\{z,y\}}=T(z,y)$.
\qed\end{proof}

\section{Complexity of Binet Compatibility}\label{sec:complexity}

A direct consequence of Theorem~\ref{thm:binet} is that there exists a simple polynomial-time algorithm to decide whether there exists a binary level-1 network displaying a given collection~$\bn$ of binary level-1 binets (see~\cite{himsw} for a related algorithm). In particular, a sink set of~$D(\bn)$ can be found in polynomial-time by computing the strongly connected components of~$D(\bn)$~\cite{tarjan1972depth} and checking for each of them whether it is a sink set. This can be used to find a typed split, if it exists. If such a split does not exist, then~$\bn$ is not compatible. Otherwise, we can try to construct networks for~$\bn|_A$ and~$\bn|_B$ recursively, and combine them as described in the proof of Theorem~\ref{thm:binet}. This algorithm is similar to the Aho algorithm for deciding whether a set of rooted trees can be displayed by some rooted tree~\cite{aho1981inferring}.

From Theorem~\ref{thm:binary}, it now follows that the following problem can also be solved in polynomial time.

\bigskip
\noindent
{\bf Binet Compatibility} (BC) \\
\noindent
{\bf Input:} a set $\bn$ of binary level-1 binets.\\
\noindent
{\bf Question:} is $\bn$ compatible, i.e., does there exist a binary network $N$ with $\bn\subseteq \bn(N)$?\\


We show now that the assumption that all binets in~$\bn$ are binary and level-1 is essential. Indeed, for general binets, the compatibility problem is at least as hard as the well-known graph isomorphism problem (GI)~\cite{GI1,GI2}, which is not known to be solvable in polynomial time. This is even true when the given binet set is thin (contains at most one binet for each pair of leaves).

\begin{theorem}\label{thm:gi}
Deciding whether there exists a phylogenetic network displaying a given thin set $\bn$ of binets is GI-hard.
\end{theorem}
\begin{proof}
We reduce from DAG-isomorphism, which is known to be GI-complete~\cite{GI2}. Let~$G_1,G_2$ be two directed acyclic graphs, which form an instance of the DAG-isomorphism problem. For~$i=1,2$, we add vertices~$\rho_i,u_i,v_i,w_i,r_i$, a new leaf labelled~$x$, an arc from~$w_i$ to each indegree-0 vertex of~$G_i$ and from each outdegree-0 vertex of~$G_i$ to~$r_i$ and arcs $(\rho_i,u_i)$, $(u_i,v_i)$, $(\rho_i,v_i)$, $(v_i,w_i)$ and~$(r_i,x)$. In~$G_1$, we add a new leaf labelled~$y$ and an arc~$(u_1,y)$. In~$G_2$, we add a new leaf labelled~$z$ and an arc~$(u_2,z)$. We have thus transformed~$G_1$ into a binet~$B_1$ and~$G_2$ into a binet~$B_2$. The third binet is~$B_3=T(y,z)$. \rev{See Figure~\ref{fig:reduction} for an illustration.}

\begin{figure}
\centering
 	\begin{tikzpicture}
	\begin{scope}[xshift=0cm,yshift=0cm]
		\draw [fill] (0,0) circle (0.05);
		\draw [fill] (.5,-.5) circle (0.05);
		\draw [fill] (0,-1) circle (0.05);
		\draw [fill] (1,-1) circle (0.05) node[below right] {$y$};
		\draw [fill] (0,-1.5) circle (0.05);
		\draw[thick] (0,0) -- (.5,-.5);
		\draw[thick] (0,0) -- (0,-1);
		\draw[thick] (.5,-.5) -- (0,-1);
		\draw[thick] (.5,-.5) -- (1,-1);
		\draw[thick] (0,-1) -- (0,-1.5);
		\draw[thick] (0,-1.5) -- (-0.5,-2);
		\draw[thick] (0,-1.5) -- (0,-2);
		\draw[thick] (0,-1.5) -- (0.5,-2);
		\draw (0,-2.25) node {$G_1$};
		\draw[thick] (0,-2.5) -- (0,-3);
		\draw[thick] (-.5,-2.5) -- (0,-3);
		\draw[thick] (.5,-2.5) -- (0,-3);
		\draw [fill] (0,-3) circle (0.05);
		\draw [fill] (0,-3.5) circle (0.05) node[below] {$x$};
		\draw[thick] (0,-3) -- (0,-3.5);
		\draw (0,-4.5) node {$B_1$};
	\end{scope}
		\begin{scope}[xshift=3cm,yshift=0cm]
		\draw [fill] (0,0) circle (0.05);
		\draw [fill] (.5,-.5) circle (0.05);
		\draw [fill] (0,-1) circle (0.05);
		\draw [fill] (1,-1) circle (0.05) node[below right] {$z$};
		\draw [fill] (0,-1.5) circle (0.05);
		\draw[thick] (0,0) -- (.5,-.5);
		\draw[thick] (0,0) -- (0,-1);
		\draw[thick] (.5,-.5) -- (0,-1);
		\draw[thick] (.5,-.5) -- (1,-1);
		\draw[thick] (0,-1) -- (0,-1.5);
		\draw[thick] (0,-1.5) -- (-0.5,-2);
		\draw[thick] (0,-1.5) -- (0,-2);
		\draw[thick] (0,-1.5) -- (0.5,-2);
		\draw (0,-2.25) node {$G_2$};
		\draw[thick] (0,-2.5) -- (0,-3);
		\draw[thick] (-.5,-2.5) -- (0,-3);
		\draw[thick] (.5,-2.5) -- (0,-3);
		\draw [fill] (0,-3) circle (0.05);
		\draw [fill] (0,-3.5) circle (0.05) node[below] {$x$};
		\draw[thick] (0,-3) -- (0,-3.5);
		\draw (0,-4.5) node {$B_2$};
	\end{scope}
		\begin{scope}[xshift=6cm,yshift=-1cm]
		\draw [fill] (0,0) circle (0.05);
		\draw [fill] (-.5,-.5) circle (0.05) node[below left] {$z$};
		\draw [fill] (.5,-.5) circle (0.05) node[below right] {$y$};
		\draw[thick] (0,0) -- (-.5,-.5);
		\draw[thick] (0,0) -- (.5,-.5);
		\draw (0,-2) node {$B_3$};
	\end{scope}
	\begin{scope}[xshift=9cm,yshift=0cm]
		\draw [fill] (0,0) circle (0.05);
		\draw [fill] (.5,-.5) circle (0.05);
		\draw [fill] (0,-1) circle (0.05);
		\draw [fill] (1,-1) circle (0.05);
		\draw [fill] (.5,-1.5) circle (0.05) node[below right] {$z$};
		\draw [fill] (1.5,-1.5) circle (0.05) node[below right] {$y$};
		\draw[thick] (1,-1) -- (1.5,-1.5);
		\draw[thick] (1,-1) -- (.5,-1.5);
		\draw [fill] (0,-1.5) circle (0.05);
		\draw[thick] (0,0) -- (.5,-.5);
		\draw[thick] (0,0) -- (0,-1);
		\draw[thick] (.5,-.5) -- (0,-1);
		\draw[thick] (.5,-.5) -- (1,-1);
		\draw[thick] (0,-1) -- (0,-1.5);
		\draw[thick] (0,-1.5) -- (-0.5,-2);
		\draw[thick] (0,-1.5) -- (0,-2);
		\draw[thick] (0,-1.5) -- (0.5,-2);
		\draw (0,-2.25) node {$G_1$};
		\draw[thick] (0,-2.5) -- (0,-3);
		\draw[thick] (-.5,-2.5) -- (0,-3);
		\draw[thick] (.5,-2.5) -- (0,-3);
		\draw [fill] (0,-3) circle (0.05);
		\draw [fill] (0,-3.5) circle (0.05) node[below] {$x$};
		\draw[thick] (0,-3) -- (0,-3.5);
		\draw (0,-4.5) node {$N$};
	\end{scope}
	\end{tikzpicture}
\caption{\rev{The three binets constructed in the proof of Theorem~\ref{thm:gi}, and a network~$N$ that displays all three binets if the digraphs $G_1$ and~$G_2$ are isomorphic.}}
\label{fig:reduction}
\end{figure}
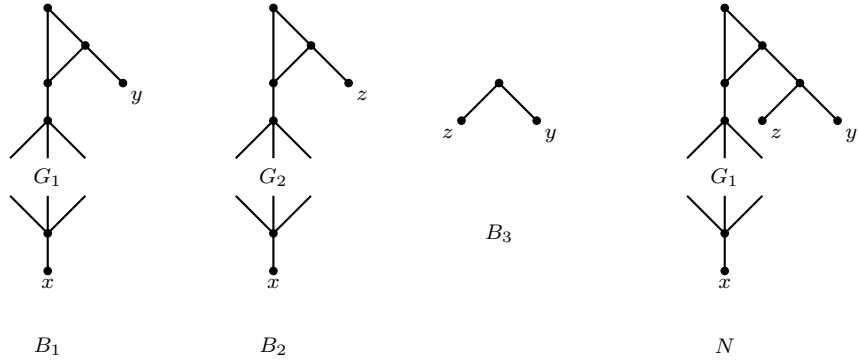

We claim that~$G_1$ and~$G_2$ are isomorphic if and only if there exists a network displaying~$B_1,B_2$ and~$B_3$.

First assume that~$G_1$ and~$G_2$ are isomorphic. Then we can construct a network displaying~$B_1,B_2$ and~$B_3$ as follows. Take~$B_1$ and subdivide the arc~$(u_1,y)$ by a new vertex~$u_1'$ and add leaf~$z$ with an arc~$(u_1',z)$. The obtained network clearly displays~$B_1$ and~$B_3$ and it also displays~$B_2$ since~$G_1$ and~$G_2$ are isomorphic.

Now assume that there exists some network~$N$ displaying~$B_1,B_2$ and~$B_3$. Then $N|_{\{x,y\}}=B_1$. Hence,~$N$ contains a cycle (in the underlying undirected graph) containing a reticulation~$v$, such that~$x$ and the image of~$G_1$ are below the arc leaving~$v$, while~$y$ is below some other arc leaving the cycle. Since~$N|_{\{y,z\}}=T(y,z)$, leaf~$z$ is not below~$v$ in~$N$. Therefore, deleting~$v$,~$x$ and the parent of~$x$ from the subgraph of~$N$ rooted at~$v$ gives~$G_1$.

Similarly,~$N$ contains a cycle containing a reticulation~$v'$, such that~$x$ and the image of~$G_2$ are below the arc leaving~$v'$, while~$z$ is below some other arc leaving the cycle. Since~$N|_{\{y,z\}}=T(y,z)$, leaf~$y$ is not below~$v'$ in~$N$. Therefore, deleting~$v'$,~$x$ and the parent of~$x$ from the subgraph of~$N$ rooted at~$v'$ gives~$G_2$.

Moreover,~$v=v'$ since~$N|_{\{y,z\}}=T(y,z)$. Hence,~$G_1$ and~$G_2$ are isomorphic.
 \qed\end{proof}

%

%
%


\section{Maximum Binet Compatibility}\label{sec:max}

If a collection of binets is not compatible, the question arises whether \rev{it is possible to find a largest compatible subset of the binets, in polynomial time. Here we show that this is unlikely to be the case.} The decision version of this problem is defined as follows.

\bigskip
\noindent
{\bf Maximum Binet Compatibility} (MBC)\\
\noindent
{\bf Input:} a set $\bn$ of binary level-1 binets and an integer $k$.\\
\noindent
{\bf Question:} does there exist a compatible subset $\bn'$ of $\bn$ with $|\bn'| \geq k$? \\



\medskip
We now establish the complexity of this problem (see Theorem~\ref{thm:kBC:hard}). \rev{Recall from Section~\ref{sec:retic} that~$s(C)$ and~$t(C)$ denote the source and terminal of a semi-cycle~$C$, respectively.}

\begin{lemma}
\label{lem:semi-cycle}
If the binet $R(x;y)$ is displayed by a binary level-1 network $N$, then $\lsa(x,y)$ is the source of a semi-cycle $C$ in $N$. In addition, $y$ is below $t(C)$ and $x$ is not below $t(C)$.    
\end{lemma}

\begin{proof}
Let $u=\lsa(x,y)$. Note that $u$ is not a reticulation vertex, as otherwise the child of~$u$ would be a stable ancestor of $x$ and $y$ that is below~$u$. Hence,~$u$ has two children, denoted by~$u_1$ and~$u_2$.

Observe that neither $(u,u_1)$ nor $(u,u_2)$ is a cut arc, since otherwise we would have $N|_{\{x,y\}}=T(x,y)$, while by the assumption of the lemma $N|_{\{x,y\}}=R(x;y)$. Hence,~$u$ is the source of a semi-cycle~$C$.  Let $v:=t(C)$ be the terminal of~$C$. If neither $x$ nor $y$ is below $v$, then $N|_{\{x,y\}}=T(x,y)$, a contradiction. If both~$x$ and $y$ are below $v$, then $v$ is a stable ancestor of $x$ and $y$, a contradiction to $\lsa(x,y)=u$. Therefore, precisely one of~$x$ and~$y$ is below $v$. If $x$ is below $v$ and $y$ is not, then $N|_{\{x,y\}}=R(y;x)$, a contradiction. Therefore, $y$ is below $v$ and $x$ is not.
\qed\end{proof}

In view of the last lemma, for each binet $R(x;y)=N|_{\{x,y\}}$, there exists a unique semi-cycle $C_N(x;y)$ containing $\lsa(x,y)$.

\begin{lemma}
\label{lem:order}
If the two binets $R(x;y)$ and $R(y;z)$ are both displayed by a binary level-1 network $N$, then
$$
s(C_N(y;z)) \prec t(C_N(x;y)).
$$
\end{lemma}

\begin{proof}
Let $C_1=C_N(x;y)$ and $C_2=C_N(y;z)$. By Lemma~\ref{lem:semi-cycle}, $y\prec t(C_1)$ but $y$ is not below $t(C_2)$, from which we know that $C_1\not =C_2$. Since $s(C_1)$ and $s(C_2)$ are stable ancestors of $y$ in view of Lemma~\ref{lem:semi-cycle}, 
we have either $s(C_1) \prec s(C_2)$ or $s(C_2)\prec s(C_1)$ but not both.

Note that if $s(C_1)\prec s(C_2)$, then $s(C_1)\prec t(C_2)$ and hence $y\prec s(C_1)\prec t(C_2)$, a contradiction. Thus  $s(C_2)\prec s(C_1)$, from which it follows that $s(C_2)\prec t(C_1)$.  
\qed\end{proof}

Given a digraph $G$, let $\mathcal{R}(G)$ be the collection of binets $\{R(x;y)\,|\,(x,y) \in E(G)\}$ \rev{induced by} $G$. \rev{Note that $\mathcal{R}(G)$ is a binet set on~$V(G)$, i.e., the leaves of the binets in $\mathcal{R}(G)$ correspond to the vertices of~$G$.}

\begin{proposition}
\label{prop:digraph:binet}
Let~$G$ be a digraph. Then $G$ is acyclic if and only if $\mathcal{R}(G)$ is compatible.
\end{proposition}

\begin{proof}
Let $n=|X|$, with~$X$ the vertex set of~$G$. 
Suppose first that $G$ is acyclic, then there exists a topological sorting of $G$, that is, the vertices of $G$ can be ordered as $x_1,\ldots ,x_n$ so that $(x_i,x_j) \in E(G)$ implies $i<j$. Hence, the network $N_*$ in Figure~\ref{fig:comb} displays~$\mathcal{R}(G)$ \rev{since~$N_*$ displays each binet~$R(x_i;x_j)$ with $i<j$}.

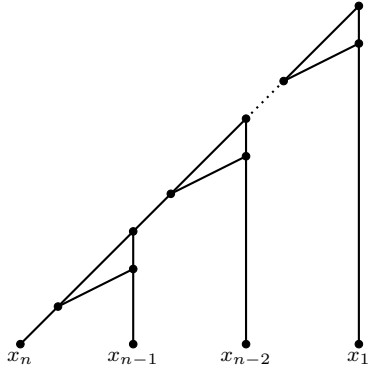
\begin{figure}
\centering
\begin{tikzpicture}

	\draw [fill] (1.5,1.5) circle (0.05);
	\draw [fill] (1.5,1) circle (0.05);
	\draw [fill] (0.5,0.5) circle (0.05);
	\draw[thick] (0.5,0.5) -- (1.5,1);
	\draw[thick] (0,0) -- (1.5,1.5);
	\draw[thick] (1.5,0) -- (1.5,1.5);
	\draw [fill] (0,0) circle (0.05) node [below] {$x_n$};
	\draw [fill] (1.5,0) circle (0.05) node [below] {$x_{n-1}$};
	\draw [fill] (3,0) circle (0.05) node [below] {$x_{n-2}$};
	\draw [fill] (3,3) circle (0.05);
	\draw [fill] (3,2.5) circle (0.05);
	\draw [fill] (2,2) circle (0.05);
	\draw[thick] (2,2) -- (3,2.5);
	\draw[thick] (1.5,1.5) -- (3,3);
	\draw[thick] (3,0) -- (3,3);
	\draw [fill] (4.5,0) circle (0.05) node [below] {$x_{1}$};
	\draw [fill] (4.5,4.5) circle (0.05);
	\draw [fill] (4.5,4) circle (0.05);
	\draw [fill] (3.5,3.5) circle (0.05);
	\draw[thick] (3.5,3.5) -- (4.5,4);
	\draw[thick,dotted] (3,3) -- (3.5,3.5);
	\draw[thick] (3.5,3.5) -- (4.5,4.5);
	\draw[thick] (4.5,0) -- (4.5,4.5);
	\end{tikzpicture}
\caption{A level-1 network $N_*$ on $X=\{x_1,\ldots,x_n\}$.}
\label{fig:comb}
\end{figure}
 
Conversely, suppose that $\mathcal{R}(G)$ is compatible. By Theorem~\ref{thm:binary},  there exists a binary level-1 network $N$ with $\mathcal{R}(G)\subseteq \bn(N)$. It remains to show that $G$ is acyclic. If not, then there exists a directed cycle $(x_1,x_2,\ldots,x_m)$ for some $m\geq 3$. Denote~$x_{m+1}=x_1$. In view of Lemma~\ref{lem:semi-cycle}, let $C_i=C_N(x_i;x_{i+1})$ be the semi-cycle in $N$ containing $\lsa(x_i,x_{i+1})$ for $1\leq i \leq m$. Then  Lemma~\ref{lem:semi-cycle} implies $x_1\prec s(C_m)$ and that $x_1$ is not below $t(C_1)$.
On the other hand, by Lemma~\ref{lem:order} we have 
$$
s(C_m) \prec t(C_{m-1}) \prec s(C_{m-1}) \prec \ldots \prec s(C_2) \prec t(C_{1}). 
$$
Together with $x_1\prec s(C_m)$, it follows that $x_1\prec t(C_1)$, a contradiction. 
\qed\end{proof}

A set of binets $\bn$ on $X$ is said to be {\em dense} if for each pair of distinct elements $x$ and $y$ in $X$, there exists precisely one binet on $\{x,y\}$ in $\bn$. \leo{Hence, a dense set of binets is always thin.}

\begin{theorem}
\label{thm:kBC:hard}
The problem MBC is NP-complete, even if the given set of binets is dense. 
\end{theorem}

\begin{proof}
We reduce from the NP-hard problem Feedback Arc Set in Tournaments (FAST)~\cite{alon2006ranking,charbit2007minimum}, which is defined as follows. Given a tournament, i.e. a digraph $G=(V,E)$ with either $(a,b)\in E$ or $(b,a)\in E$ (but not both) for each pair of distinct elements $a$ and $b$ in $V$, and given a positive integer $k'$, does there exist a subset $F\subseteq E$ of at most $k'$ arcs whose removal makes~$G$ acyclic. If \rev{such an arc} set exists, then we call it a {\em feedback arc set} of $G$.

The reduction is as follows. For each instance $(G,k')$ of FAST, consider the corresponding instance $(\mathcal{R}(G),k)$ of MBC with $k=|\mathcal{R}(G)|-k'$. Since the set $\mathcal{R}(G)$ of binets induced by $G$ can be constructed in polynomial time, it suffices to show that $G$ contains a feedback arc set with size at most $k'$ if and only if there exists a compatible subset of $\mathcal{R}(G)$ of size at least~$k$. 

First assume that there exists a feedback arc set $E'$ of $G$ with size at most $k'$. That is, $|E'|\leq k'$, and the digraph $G^*$ obtained from $G$ by deleting the arcs in $E'$ is acyclic. Consider the set of binets $\bn'=\{R(x;y)\,:\,(x,y)\in E\setminus E'\}$. This set contains at least~$k$ binets. In addition, since $\bn'=\mathcal{R}(G^*)$, it follows by Proposition~\ref{prop:digraph:binet} that $\bn'$ is compatible. 

Now assume that there exists a compatible binet set $\bn'\subseteq \mathcal{R}(G)$ with $|\bn'|\geq k$. Consider the set $E'=\{(x,y)\,:\, R(x,y)\in \mathcal{R}(G)\setminus \bn'\}$ of arcs of $G$. Then by Proposition~\ref{prop:digraph:binet}, it follows that $E'$ is a feedback arc set. Moreover, $|E'|\leq k'$, which completes the proof. \qed\end{proof}


We complete the section by showing that there exists a polynomial time $1/3$-approximation algorithm for the MBC problem, which follows directly from the next theorem and its proof.

\begin{theorem}\label{thm:questionmark}
Suppose that $\bn$ is a set of binary level-1 binets on $X$. Then there exists a binary level-1 network $N$ such that $|\bn(N)\cap \bn|\geq |\bn|/3$.
\end{theorem}

\begin{proof}
If at least a third of the binets in $\bn$ are tree type, then take $N$ to be any \rev{binary} tree on $X$ and we are done. Hence we may assume that at least two thirds of the binets are reticulate type. 

Impose an arbitrary ordering on the elements in $X$, that is, write $X=\{x_1,\ldots,x_n\}$. Let $\bn_1=\bn\cap \{R(x_i;x_j)\,:\,1\leq i<j\leq n\}$ and $\bn_2=\bn \cap \{R(x_j;x_i)\,:\,1\leq i<j\leq n\}$. Without loss of generality, we may assume that $|\bn_1|\geq |\bn_2|$ (as the other case can be established in a similar way). Since at least two thirds of the binets are reticulate type, and each of those is contained in either $\bn_1$ or $\bn_2$ (but not both), we know that $|\bn_1|\geq |\bn|/3$.
Now consider the network $N_*$ in Figure~\ref{fig:comb}, then clearly we have $\bn_1\subseteq \bn(N_*)$. Thus we have 
$|\bn(N_*)\cap \bn|\geq |\bn_1| \geq |\bn|/3$, from which the theorem follows. 
\qed\end{proof}


\section{Discussion}\label{sec:discussion}

In this paper we have developed some combinatorial results concerning collections of level-1 binets.  Several interesting questions arise from these results. For example, we have shown that the collection of level-1 binets displayed \rev{by} a binary phylogenetic network can be displayed by some level-1 network, but is there some canonical level-1 network that could be used to display such a collection? In addition, can we count the number of binary level-1 networks that display a dense compatible collection of binets? We have also seen that the collection of binets displayed by a binary level-1 network determine its reticulation number. Therefore it is natural to ask which properties of a phylogenetic network in general are determined by its binets?

We have also studied some algorithmic questions concerning binets. Concerning the maximum binet compatibilty problem, note that the constant $1/3$ is sharp in Theorem~\ref{thm:questionmark}. For example, consider the binet collection $\{R(x;y),T(x,y),R(y;x)\}$. However, can a better bound be achieved by restricting to thin collections of binets, and can improved approximation algorithms also be found?

In another direction, it would be interesting to know whether similar results to those proven in this paper might hold for higher level networks. For example, what can be said about properties of collections of level-2 binets, and does Theorem~\ref{thm:retic} hold also for higher level networks? Also, we could try to generalize some of our results to $k$-nets, i.e. networks on $k$ leaves, $k \ge 2$. For example, does Theorem~\ref{thm:binary} hold for trinets? In general, it would be interesting to know what additional information the collection of $k$-nets displayed by a network might contain for $k\ge 3$. Note that it has been shown that trinets do not completely determine rooted networks in general~\cite{huber2015much}. However, do they determine properties of networks such as the number of reticulations?
 
Similarly, it would be interesting to extend some of our algorithmic results to higher-level networks and $k$-nets. For example, it is known that the \rev{compatibility} problem is NP-complete for collections of level-1 trinets \cite{himsw}. However, to date the maximum trinet \rev{compatibility} problem has not been studied.


\rev{Eventually, it is hoped that new results in these directions could be useful for developing novel methods to construct phylogenetic networks from higher-level networks and k-nets. For example, using our results it may be possible to develop approaches to build a consensus network for a collection of phylogenetic trees or networks. Note that consensus networks have already proven themselves useful in the unrooted setting, where they are used to summarize key features displayed by a collection of trees or networks (see e.g. \cite{holland2004consensus}).  A consensus method based on binets could work by breaking each of the given networks down into a collection of binets, and then developing methods to pool together the information contained in
the resulting binets so as to construct some consensus network, or at least some constraints that any such network should satisfy. Note that similar approaches have been developed to build consensus trees for a collection of phylogenetic trees by breaking each of the trees down into a collection of triplets (see e.g. \cite[Section 2]{bryant2003classification}). Probably it would be of some interest to first consider how to construct a level-1 consensus network for a collection of level-1 networks by breaking each of them down into level-1 binets. This is already likely to be quite challenging in view of our result concerning NP-completeness of Maximum Binet Compatibility.}


\bibliographystyle{spmpsci}      
\bibliography{phynet}   


\end{document}